\documentclass[12pt]{article}
\usepackage[T2A]{fontenc}
\usepackage[cp1251]{inputenc}
\usepackage[english]{babel}
\usepackage{amssymb}
\usepackage{amsmath}
\usepackage{amsthm}
\usepackage{amsfonts}
\usepackage{amssymb}

\newtheorem{theorem}{Theorem}

\newtheorem{lemma}{Lemma}

\newtheorem{proposition}{Proposition}

\newtheorem{definition}{Definition}

\newtheorem{example}{Example}
\newtheorem{remark}{Remark}

\begin{document}

\begin{center}
\textbf{\uppercase{Levy Laplacians and instantons on manifolds}}

B.~O.~Volkov

borisvolkov1986@gmail.com

Steklov Mathematical Institute,
           
ul. Gubkina, 8, Moscow, 119991, Russia

\end{center}

Dedicated to the memory of Alexander A. Belyaev

\textbf{Abstract}: 
 The equivalence of the anti-selfduality   Yang-Mills equations  on the $4$-dimensional orientable Riemannian manifold and  Laplace equations for some  infinite dimensional Laplacians is proved. A class of modificated L\'evy Laplacians parameterized  by the choice of a curve in the group $SO(4)$ is introduced.
It is shown that a connection  is an instanton (a solution of the anti-selfduality   Yang-Mills equations)  if and only if the  parallel transport generalized by this connection is a solution of the 
Laplace equations for some three modificated  Levy Laplacians from this class.

2010 Mathematics Subject Classification: 70S15,58J35

key words: Levy Laplacian, Yang--Mills equations, instantons,  infinite-dimensional manifold 

\section*{Intoduction}
One of the main causes of the interest in  infinite dimensional differential operators of the Levy type  is their connection 
with the Yang-Mills fields.
The Yang-Mills equations on a connection $A$ in the vector bundle over $d$-dimensional  orientable Riemannian manifold $M$  are
\begin{equation}
\label{ymint}
D_A^\ast F=0,
\end{equation}
where $F=dA+A\wedge A$ is the curvature of $A$  and $D_A^\ast$ is the adjoint operator to the exterior covariant derivative. The Yang-Mills fields are solutions of the Yang-Mills equations.
The parallel transport $U$ generated by the connection $A$ can been considered as a section in some vector bundle over
the Hilbert manifold of $H^1$-curves with the fixed origin  in $M$ (if $M=\mathbb{R}^d$ the parallel transport can been considered  as an operator-valued function  on  the  Hilbert space  of $H^1$-curves with the origin at zero).
The theorem proved by Accardi, Gibilisco and Volovich in~\cite{AGV1994} states that the connection $A$ in a vector bundle over $\mathbb{R}^d$ satisfies the Yang-Mills equations if and only if
$$\Delta_L^{AGV}U=0,$$ where $\Delta_L^{AGV}$ is some infinite dimensional Laplacian. This Laplacian was defined by analogy with the famous Levy Laplacian  $\Delta_L$ (see~\cite{L1951}) and was also called the same name.  Accardi-Gibilisco-Volovich  theorem was generalized for Riemannian manifolds by Leandre and Volovich in~\cite{LV2001}.

In the case $d=4$, the Hodge dual $\ast$ transforms  2-forms on $M$ into 2-forms. So it is possible to consider 
the selfduality equations 
\begin{equation}
\label{dualintro}
F=\ast F
\end{equation} or anti-selfduality   equations 
\begin{equation}
\label{antidualintro}
F=-\ast F
\end{equation}
 on a connection $A$.
A connection
 is called  an instanton or  an anti-instanton if it is a solution of equations~(\ref{antidualintro}) or~(\ref{dualintro}) respectively (see~\cite{Sergeev}). Any connection $A$ satisfies the Bianchi identities $D_AF=0$. Due to $D_A^\ast=-\ast D_A \ast$,  instantons and
anti-instantons are solutions of the Yang--Mills equations~(\ref{ymint}). In the current paper,
the family of modificated  Levy Laplacians  is introduced.  It is shown that the connection satisfies the anti-selfduality Yang-Mills equations on
 a $4$-dimensional orientable Riemannian manifold if and only if the  parallel transport satisfies the Laplace equations for three operators from this family. In the fact, the problem of the description of instantons in the terms of the parallel transport
 and the Levy Laplacians was stated  by Accardi in~\cite{Accardi1995}. So, in the current paper, this problem is solved for the Riemannian case. For the flat case it was solved by author in~\cite{VolkovLLI}.

 The following scheme from~\cite{ASF} can been useful for the definition of  differential operators 
particularly for the definition of the modificated Levy Laplacians.
Let $E$ be a real normed vector space and $E^\ast$ be its conjugate. Let $\mathcal{L}(E,E^\ast)$ be the space of  all linear continuous operators from $E$ to $E^\ast$. 
If $f\in C^2(E,\mathbb{R})$, then $f'(x)\in E^\ast$ and $f''(x)\in \mathcal{L}(E,E^\ast)$ for any $x\in E$.
 Let  $S\colon dom\,S\to \mathbb{R}$ be a linear functional and  $dom\,S\subset \mathcal{L}(E,E^\ast)$.  The functional $S$ defines the second order differential operator $D^{2,S}$ by the formula
 \begin{equation}
D^{2,S}f(x)=S(f''(x)).
\end{equation}
If we choose $E=\mathbb{R}^d$ and $S=tr$, then
 $D^{2,tr}$ is the Laplace operator $\Delta$. 
The original Levy Laplacian on the functions on  $L_2[0,1]$  was introduced by P.~Levy (see~\cite{L1951}). It can been 
defined   as the second order differential operator $D^{2,tr_L}$ associated with  the so called  the Levy trace  $tr_L$  (see~\cite{AS1993}). 
The Levy trace is a linear functional defined in the following way.
Let  $\mathcal{K}_{comp}$  be the ring of compact operators on $L_2[0,1]$ and $\mathcal{A}_{mult}$ be the algebra  of operators of multiplication on functions from
$L_{\infty}[0,1]$.
Let $\mathcal {A}=\mathcal{K}_{comp}\oplus \mathcal{A}_{mult}$.
We will identify  $h\in L_{\infty}[0,1]$ with the operator of multiplication  on $h$.
Then the L\'evy trace $tr_L$ is a linear functional on  $\mathcal A$ defined by
\begin{equation}
\label{10001}
tr_L K=0,
\end{equation}
if $K\in\mathcal{K}_{comp}$, 
and by
\begin{equation}
\label{10002}
tr_L h=\int_0^1h(t)dt,
\end{equation}
if $h\in L_{\infty}[0,1]$.~\footnote{This  linear functional defines a singular quantum  state and is well studied (see for example~\cite{SV,SV2}).}

The Levy Laplacian $\Delta^{AGV}_L$  that was introduced in the papers~\cite{AGV1993, AGV1994} by Accardi, Gibilisco and Volovich can been  associated
with the linear functional $tr_L^{AGV}$  on some special space  of bilinear forms on $H^1([0,1],\mathbb{R}^d)$ (see also~\cite{Volkov2019}). This linear functional is more complicated analogue of the Levy trace $tr_L$ and we will call it the same name.
 The tangent bundle over the Hilbert manifold of the $H^1$-curves with the fixed origin  in the Riemannian manifold $M$ is trivial. It allows to transfer  the scheme of the definition of the second order differential operators  to the space  of sections in the vector bundle over this Hilbert manifold. In this case, the linear functional $tr_L^{AGV}$ defines the Levy Laplacian that was used by Leandre and Volovich  in~\cite{LV2001}.

Any smooth curve $W\in C^1([0,1],SO(4))$ defines an orthogonal operator 
in $L_2([0,1],\mathbb{R}^4)$ by pointwise left multiplication. The subspace  $H^1([0,1],\mathbb{R}^4)\subset L_2([0,1],\mathbb{R}^4)$ is invariant under the action of $W$.  The modificated Levy trace associated with $W\in C^1([0,1],SO(4))$ acts on billinear form $K$ on  $H^1([0,1],\mathbb{R}^4)$ by formula
\begin{equation*}
tr^W_L K=tr_L^{AGV}W^\ast K W.
\end{equation*}
If $W$ is not constant, the modificated Levy trace $tr^W_L$ does not coincide with the Levy trace  $tr_L^{AGV}$. So the Levy trace has not some properties of an usual trace. The modificated Levy Laplacian $\Delta^W_L$ is the second order derivative operator  associated with $tr_L^{W}$.

The  group $SO(4)$ is not simple and has the normal subgroups $S^3_L\cong SU(2)$ and $S^3_R\cong SU(2)$.
The Lie algebra $so(4)=Lie(S^3_L)\oplus Lie(S^3_R)$  and this
corresponds to decomposition of the space of 2-forms into the direct sum of the space of selfdual and anti-selfdual 2-forms.
In the paper, we show that a connection $A$ is an instanton  (antiinstanton) on a $4$-dimensional orientable Riemannian manifold  if and only if $\Delta^W_LU=0$ for any $W\in C^1([0,1],S^3_L)$ (for any $W\in C^1([0,1],S^3_R)$).
Let $\{\bf{e_1},\bf{e_2},\bf{e_3}\}$ be a some  basis of the Lie algebra of  $S^3_L$.
Let $W_i(t)=e^{t\bf{e_i}}$ for $i\in \{1,2,3\}$.
We prove that it is sufficient to check $\Delta^{W_i}_LU_{1,0}=0$ to show that the connection $A$ is an instanton.

The modificated Levy Laplacians were introduced in the work~\cite{VolkovLLI} by author, where only the flat case and instantons over $\mathbb{R}^4$ were considered. In that case,  it is possible to use only one Laplace equation
for some modificated  Levy Laplacian instead of three of them. In~\cite{VolkovLLI},   the sufficient conditions on a smooth curve $W\in C^1([0,1],S^3_R)$
that 
the equality  $\Delta_L^WU=0$ implies that the connection $A$ is an instanton were found. 
In the proof, the fact  that $\mathbb{R}^4$  is not compact was essentially used. So it is the open question whether it is possible 
 to transfer the result of~\cite{VolkovLLI} for an arbitrary    4-dimensional orientable Riemannian  manifold.
The simple Abelian case for a 4-dimensional orientable Riemannian  manifold was considered in~\cite{Volkov2016}.

Another approach to the definition of the Levy Laplacian is to define it  as the Cesaro mean of the second order directional derivatives 
along the vectors of some orthonormal basis (see~\cite{L1951,KOS}).
This approach can be also useful in the connection with the Yang-Mills equations (see~\cite{Volkovdiss,VolkovD,Volkov2017,Volkov2019,Volkov2019a}) and instantons (see~\cite{VolkovLLI,VolkovVINITI}).
Different approaches to the Yang-Mills equations  based on the parallel transport but not based on the Levy Laplacian
were used in~\cite{Gross,Driver,Bauer1998,Bauer2003,ABT,AT}. Particularly,  instantons were studied in~\cite{Bauer1998,Bauer2003,ABT}.

The paper is organized as follows. In Sec.~\ref{Sec1}, we give preliminary information about the Yang-Mills equations and instantons on  4-dimensional orientable Riemannian   manifolds. In Sec.~\ref{Sec2}, we give preliminary information about the parallel transport.  We consider it as  a section in some infinite-dimensional vector bundle over the Hilbert manifold of $H^1$-curves with the fixed origin.
In Sec.~\ref{Sec3}, we transfer the scheme of  the definition of the second order derivative operators on  the space of sections in this infinite dimensional bundle. In Sec.~\ref{Sec4}, we define the modificated Levy trace as the result of the action of the curve from $C^1([0,1],SO(4))$ on the Levy trace. We define the modificated  Levy Laplacian as the second order derivative operator associated with the modificated Levy trace. In Sec.~\ref{Sec5},  we find the value of the modificated Levy Laplacian on the parallel 
transport. In Sec.~\ref{Sec6}, we prove the main theorem on the equivalence of the self-duality   Yang-Mills equations and 
the Laplace equations for the modificated Levy Laplacians.

\section{Instantons on manifold}
\label{Sec1}
In the paper, all manifolds are finite or  infinite dimensional Hilbert manifolds. 
In the infinite-dimensional case, all derivatives are  understood in the Frechet sense and the  symbol $d_X$ will denote the derivative in the  direction $X$.
For  information about  the infinite  dimensional geometry see~\cite{Lang,Klingenberg,Klingenberg2}.

Let $M$ be a smooth orientable Riemannian 4-dimensional manifold. Let $g$ denote the Riemannian metric on $M$.
We will raise and lower indices using this  metric  and we will sum over repeated indices.
Let $G\subseteq  SU(N)$ be a closed Lie group and $Lie (G)\subseteq su(N)$ be its Lie algebra.
Let $E=E(\mathbb{C}^N,\pi,M,G)$  be a vector bundle over $M$ with the projection $\pi\colon E\to M$, the fiber $\mathbb{C}^N$ 
and the structure group $G$. We will denote the fiber  $\pi^{-1}(x)\cong\mathbb{C}^N$ over $x\in M$ by the symbol $E_x$. Let $P$ be the principle bundle over $M$  associated with $E$.
 Let $ad (P)=Lie(G)\times_G M$ and $aut P=G\times_G M$ be the adjoint and automorphism bundles of $P$ respectively (the fiber of $adP$ is
  isomorphic to $Lie (G)$ and the fiver of $aut P$ is isomorphic to $G$)

A connection $A$ in the vector bundle $E$ is  a smooth section in $\Lambda^1\otimes adP$. (The symbol $\Lambda^p$ denotes the bundle of exterior
$p$-forms.)
If $W_a$ is an open subset of $M$ and  $\psi_a\colon \pi^{-1}(W_a)=W_a\times \mathbb{C}^N$ 
is a local trivialization of $E$, then, in this local trivialization, the connection $A$  is a smooth $Lie(G)$-valued 1-form $A^a(x)=A^a_\mu(x)dx^\mu=\psi_{a}A(x)\psi_{a}^{-1}$ on $W_a$.
Let $\psi_a\colon\pi^{-1}(W_a)\cong W_a\times \mathbb{C}^N$ and $\psi_b\colon\pi^{-1}(W_b)\cong W_b\times \mathbb{C}^N$ be two local trivializations of $E$ and $\psi_{ab}\colon W_a\cap W_b\to G$ be the transition function, i.e. it is  the function such that $\psi_{a}\circ \psi_{b}^{-1}(x,\xi)=(x,\psi_{a b}(x)\xi)$ for all $(x,\xi)\in (W_a\cap W_b)\times \mathbb{C}^N$.
If $x\in W_a\cap W_b$, then
\begin{equation}
\label{connection}
A^b(x)=\psi_{ab}^{-1}(x)A^a(x)\psi_{ab}(x)+\psi_{ab}^{-1}(x)d\psi_{ab}(x).
\end{equation}
 The curvature $F$  of the connection $A$ is a smooth section in $\Lambda^2\otimes ad P$ such that, in the local trivialization, it has the form  $F^a(x)=\sum_{\mu<\nu} F^a_{\mu\nu}(x)dx^\mu\wedge dx^\nu=\psi_{a}F(x)\psi_{a}^{-1}$,
where  $F^a_{\mu\nu}=\partial_\mu A^a_\nu-\partial_\nu A^a_\mu+[A^a_\mu,A^a_\nu]$.
If $x\in W_a\cap W_b$, then
\begin{equation}
\label{curvature}
F^b(x)=\psi_{ab}^{-1}(x)F^a(x)\psi_{ab}(x).
\end{equation}
 If $\phi$ is a smooth section in $ad P$, its
covariant derivative is defined as
$$
\nabla_\mu\phi=\partial_\mu\phi+[A_\mu,\phi].
$$
Also the following holds
$$
(\nabla_\mu\nabla_\nu-\nabla_\nu\nabla_\mu)\phi=[F_{\mu\nu},\phi].
$$

Let $D_A\colon C^\infty(M,\Lambda^p\otimes adP)\to C^\infty(M,\Lambda^{p+1}\otimes adP)$ be the operator of the exterior covariant derivative. 
It is determinated by its action on  forms $\alpha\otimes \phi$, where $\alpha$ is  a real $p$-form and
$\phi$ is a section in $adP$, by the formula
$$
D_A(\alpha\otimes \phi)=d\alpha\otimes \phi+(-1)^{p}\alpha\otimes \nabla \phi,
$$
where $d$ denotes the operator of the usual exterior derivative.
Let $D_A^{\ast}:C^\infty(M,\Lambda^{p+1}\otimes adP) \to C^\infty(M,\Lambda^{p}\otimes adP)$ be a formally adjoint to the operator $D_A$.
We have $D_A^{\ast}=-\ast D_A\ast$, where $\ast$ is the Hodge star on the manifold $M$.

The Yang--Mills action functional has the form
\begin{equation}
\label{YMaction1}
S_{YM}(A)=-\frac 12\int_{M}tr(F_{\mu\nu}(x)F^{\mu\nu}(x))Vol(dx),
\end{equation}
where $Vol$ is the Riemannian volume measure on the manifold $M$.
The Yang--Mills equations  on a connection $A$ have the form
\begin{equation}
\label{YMequations}
(D_A^\ast F)=0.
\end{equation}
In local coordinates, we have
$$
(D_A^\ast F)_\nu=-\nabla^\mu F_{\mu\nu}
$$
and
 \begin{equation*}
 \nabla_\lambda F_{\mu
\nu}=\partial_\lambda F_{\mu \nu}+[A_\lambda,F_{\mu
\nu}]-F_{\mu
\kappa}\Gamma^\kappa_{\lambda\nu}-F_{\kappa\nu}\Gamma^\kappa_{\lambda\mu},
\end{equation*}
where $\Gamma^\kappa_{\lambda\nu}$  are the Christoffel symbols of the Levy-Civita connection on $M$.
The Yang--Mills equations are the Euler-Lagrange equations for the Yang--Mills action functional~(\ref{YMaction1}).

The Hodge star acts on the curvature in the following way. If $\varepsilon_{\mu\nu\lambda\kappa}$ is the Levi-Civita symbol, then $(*F)_{\mu\nu}=\frac{\sqrt{|\det g|}}{2}\varepsilon_{\mu\nu\lambda\kappa}F^{\lambda\kappa}$.
The selfduality (anti-selfduality equations) are following equations on the connection $A$:
\begin{equation} 
\label{autodual}
F=\ast  F \,\,(F=-\ast F).
\end{equation}
 Let $F_-=F-\ast F$ and $F_+=F+\ast F$ be anti-selfdual and selfdual parts of the curvature $F$ respectively.
 The selfduality (anti-selfduality equations) can been rewritten
 \begin{equation} 
\label{autodual}
F_-=0\,(F_+=0).
\end{equation}
If a connection is a solution of the  self-duality equations or the antiself-duality equations than it is called the antiinstanton
or the instanton respectively. 
The instantons and the antiinstantons are local extrema of the Yang--Mills action functional~(\ref{YMaction1}).

The gauge transform is a smooth section in $Aut P$. Such a section $\psi$ acts on the connection by the formula
\begin{equation}
A\to A'=\psi^{-1}A\psi+\psi^{-1}d\psi
\end{equation}
and on the curvature by the formula
\begin{equation}
F\to F'=\psi^{-1}F\psi
\end{equation}
The Lagrange function of~(\ref{YMaction1}), the Yang--Mills equations~(\ref{YMequations}), the self-duality equations and the antiself-duality equations are invariant under the action of gauge transform.
The moduli space of instantons is the factor space 
 of all instantons  with the respect to the gauge equivalence. 
The moduli space of instantons over  $\mathbb{R}^4$  was described in~\cite{ADHM}. The moduli space of instantons over a 4-dimensional oriented Riemannian compact manifold was described in~\cite{Taubes}. If the intersection form on the manifold is positive then on this manifold  there exist  solutions of the self-dual Yang--Mills equations and the moduli space of instantons is a 5-dimensional manifold (see also~\cite{FU}). In the case of the self-dual base manifold instantons were described in~\cite{AHS}.
The review on the gauge fields and the instantons can been found in~\cite{Sergeev}.

\section{Parallel transport}  
\label{Sec2}

For any sub-interval $I\subset[0,1]$ let the symbol   $H^1(I,\mathbb{R}^4)$ denote the space of
 all $H^1$-functions  on $I$ with values in $\mathbb{R}^4$. 
 It is the Hilbert space with scalar product
   $$(h_1,h_2)_1=\int_I(h_1(t),h_2(t))_{\mathbb{R}^4}dt+\int_I(\dot{h}_1(t),\dot{h}_2(t))_{\mathbb{R}^4}dt.$$
Let $H_0^1=\{h\in H^1([0,1],\mathbb{R}^4)\colon h(0)=0)$
and $H_{0,0}^1=\{h\in H_0^1\colon h(1)=0\}$.

The curve $\gamma\colon [0,1]\to M$  on the manifold $M$ is called $H^1$-curve, if $\phi_a\circ \gamma\mid_I\in H^1(I,\mathbb{R}^4)$
 for any interval  $I\subset [0,1]$ and for any  coordinate chart
 $(\phi_a, W_a)$ of the manifold  $M$ such that  $\gamma(I)\subset W_a$.
 Let $\Omega$ be the set of all $H^1$-curves in $M$. If $m\in M$ let $\Omega_m=\{\gamma\in \Omega\colon \gamma(0)=m\}$. So $\Omega_m$ is the set of all $H^1$-curves in $M$ with the origin at $m\in M$.
The sets $\Omega$ and  $\Omega_m$ can be endowed with the structure of an infinite dimensional Hilbert manifold  (see~\cite{Driver,Klingenberg,Klingenberg2,Volkov2019a}).

 Let $\mathcal E$ and $\mathcal E_m$ be the vector bundles over  $\Omega$ and $\Omega_m$ respectively, which  fiber  over $\gamma\in \Omega$ (over $\gamma\in \Omega_m$) is the space $\mathcal{L}(E_{\gamma(0)},E_{\gamma(1)})$ of all linear mappings from $E_{\gamma(0)}$ to $E_{\gamma(1)}$.  The parallel 
  transport  generated by the connection $A$ in $E$ can been considered as a section in  $\mathcal E$. 
 Let $\psi_a\colon \pi^{-1}(W_a)\cong W_a\times \mathbb{C}^N$ be a local trivialization of the vector bundle $E$. For $\gamma\in \Omega$ such that $\gamma([s_0,t_0])\subset  W_a$ let $U^a_{t,s}$, where $s_0\leq s\leq t\leq t_0$,  be a solution of the system of differential equations
\begin{equation}
\label{partransp1}
 \left\{
\begin{aligned}
 \frac d{dt}U^a_{t,s}(\gamma)=-A^a_\mu(\gamma(t))\dot{\gamma}^\mu(t)U^a_{t,s}(\gamma)\\
  \frac d{ds}U^a_{t,s}(\gamma)=U^a_{t,s}(\gamma)A^a_\mu(\gamma(s))\dot{\gamma}^\mu(s)\\
\left.U^a_{t,s}\right|_{t=s}=Id.
\end{aligned}
\right.
\end{equation}
Then $U_{t_0,s_0}(\gamma)=\psi^{-1}_a U^a_{t_0,s_0}(\gamma)\psi_a$ is the parallel transport along  the restriction of $\gamma$ on $[s_0,t_0]$.
If $\gamma([s_0,t_0])\in W_a\cap W_b$,  then equality~(\ref{connection})
implies that 
\begin{equation}
\label{partransp}
U^a_{t_0,s_0}(\gamma)=\psi_{ab}(\gamma(t_0))U^b_{t_0,s_0}(\gamma)\psi_{ba}(\gamma(s_0)).
\end{equation}
For an arbitrary $\gamma\in \Omega$ we can consider a partition $s_0=t_1\leq t_2\leq\ldots t_n=t_0$ such that
$\gamma([t_i,t_{i+1}])\subset W_{a_i}$  and a family of local trivializations $\psi_{a_i} \colon  \pi^{-1}(W_{a_i})\cong   W_{a_i}\times \mathbb{C}^N$ of the vector bundle  $E$.
 Let
\begin{equation}
\label{paraltrans}
U^{a_{n},a_{1}}_{t_0,s_0}(\gamma)=U_{t_{n},t_{n-1}}^{a_{n-1}}
(\gamma)\psi_{a_n a_{n-1}}(\gamma(t_{n-1}))\ldots
U_{t_{3},t_{2}}^{a_{2}}(\gamma)\psi_{a_{2}a_{1}}(\gamma(t_{2}))U_{t_{2},t_{1}}^{a_{1}}(\gamma).
\end{equation}
then $U_{t_0,s_0}(\gamma)=\psi^{-1}_{a_n} U^{a_n,a_1}_{t_0,s_0}(\gamma)\psi_{a_1}$. The parallel transport along $\gamma$ is $U_{1,0}(\gamma)$ .
By~(\ref{partransp}), the definition of parallel transport does not depend on the choice of the partition and the choice 
of the family of trivializations.

The parallel transport has the following properties:
\begin{enumerate}
\item The mapping $\Omega\ni \gamma\to U_{1,0}(\gamma)$ is a $C^{\infty}$-smooth section in the vector bundle $\mathcal E$ (for the proof of smoothness see~\cite{Gross,Driver})  .
\item The parallel transport does not depend on the choice of parametrization of the curve.
Let $\sigma\colon[0,1]\to [0,1]$ be a
non-decreasing piecewise $C^1$-smooth function such that $\sigma(0)=0$ and $\sigma(1)=1$. 
Then
\begin{equation}
\label{reparametr}
U_{\sigma(t),\sigma(s)}(\gamma)=U_{t,s}(\gamma\circ\sigma)
\end{equation}
for any $\gamma\in \Omega$.
\item For any $\gamma \in \Omega$ the  parallel transport satisfies the multiplicative property:
\begin{equation}
\label{group}
U_{t,s}(\gamma)U_{s,r}(\gamma)=U_{t,r}(\gamma)\text{ for $r\leq s\leq t$}.
\end{equation}
\item If the restriction of $\gamma\in \Omega$ on $[s,t]$ is constant, then
\begin{equation}
\label{indent}
U_{t,s}(\gamma)\equiv Id.
\end{equation}
\end{enumerate}

\section{Second order directional operators}
\label{Sec3}

The mapping  $X\colon [0,1]\to TM$ such that   $X(t)\in T_{\gamma(t)}M$ for any $t\in[0,1]$ is a vector field along  $\gamma\in \Omega_m$, i.e. it is a section in the pullback bundle $\gamma^{*}TM$.  If $X$ is an absolutely continuous field along  $\gamma\in \Omega$, 
its covariant derivative $\nabla X$ with respect to the Levi-Civita connection is the field along $\gamma$
defined by 
\begin{equation}
\label{covder}
\nabla X(t)= \dot{X}(t)+\Gamma(\gamma(t))(X(t),\dot{\gamma}(t)),
\end{equation}
where $(\Gamma(x)(X,Y))^\mu=\Gamma^\mu_{\lambda \nu}(x)X^\lambda Y^\nu$ in local coordinates.
Let $Q(\gamma)$ denote the parallel transport generated by the Levi-Civita connection  along the curve  $\gamma$. 
Then
$$
\nabla{X}(t)=Q_{t,0}(\gamma)\frac d{dt}(Q_{t,0}(\gamma)^{-1}X(t)).
$$
The symbol 
$H^1_\gamma(TM)$ denotes the Hilbert space of all $H^1$-fields $X$ along  $\gamma$ such that $X(0)=0$. The scalar product 
on this space is defined by the formula
\begin{equation}
\label{metrG1}
G_1(X,Y)=\int_0^1g(X(t),Y(t))dt
\\+\int_0^1g(\nabla X(t),\nabla Y(t))dt.
\end{equation}
We can identify the  Hilbert spaces $H^1([0,1],\mathbb{R}^4)$ and $H^1([0,1],T_mM)$. 
Let $\{Z_1,\ldots,Z_4\}$ be an orthonormal basis in $T_mM$.
We identify
$$
H^1([0,1],\mathbb{R}^4)\ni h(\cdot)=(h^\mu(\cdot))\leftrightarrow Z_\mu h^\mu(\cdot)\in H^1([0,1],T_mM)
$$
Due to~(\ref{covder}), for any $\gamma\in \Omega_m$ the Levy-Civita connection generates the canonical isometrical isomorphism  between
$H^1_0$ and $H^1_\gamma(TM)$, which action on $h\in H^1_0$ we will denote by $\widetilde{h}$. This isomorphism acts   by the formula
\begin{equation}
\label{isomorhism}
\widetilde{h}(\gamma;t)=Q_{t,0}(\gamma)h(t)=Z_\mu(\gamma,t)h^\mu(t),
\end{equation}
where $Z_i(\gamma,t)=Q_{t,0}(\gamma)Z_i$ for $i\in \{1,2,3,4\}$. Sometimes we will miss in the notation the dependence of the infinite dimensional field on $\gamma$.
Let  $\mathcal H^1_0$ be the vector bundle over $\Omega_m$ which fiber over $\gamma\in \Omega_m$ is $H^1_\gamma(TM)$.
 Let $\mathcal H^1_{0,0}$ denote the sub-bundle of $\mathcal H^1_0$ such that
 the  fiber of  $\mathcal H^1_{0,0}$  over  $\gamma\in\Omega_m$ is the space   $\{X\in H^1_\gamma(TM)\colon X(1)=0\}$.
The vector bundle $\mathcal H^1_0$ is the tangent bundle over $\Omega_m$. Due to isomorphism~(\ref{isomorhism}), this bundle is trivial. 
For any smooth section $f$ in  $\mathcal E_m$ there exists the section  $\widetilde{D}f$ in  $\mathcal E_m \otimes H^1_{0,0}\cong \mathcal H^1_{0,0}$ such that
$d_{\widetilde{h}}f(\gamma)=<\widetilde{D}f(\gamma),h>$ for any $h\in H^1_{0,0}$. Also there exists the section  $D^2f$ in $\mathcal E_m \otimes \mathcal{L}(H^1_{0,0},H^1_{0,0})$ such that
$<d_{\widetilde{h_1}}\widetilde{D}f(\gamma),h_2>=<\widetilde{D}^2f(\gamma)h_1,h_2>$ for any $h_1,h_2\in H^1_{0,0}$. 

The scheme of the definition of the second order differential operator can been transferred to the case of manifold in the following way.
\begin{definition}
Let $S$ be  a linear functional  on $domS\subset  \mathcal{L}(H^1_{0,0}, H^1_{0,0})$.
The domain of the second order differential operator $D^{2,S}$ associated with $S$   is the space of  all smooth sections $f$ in $\mathcal E_m$ such that $\widetilde{D}^2f(\gamma)\in domS$ for all $\gamma\in \Omega_m$. The second order differential operator $D^{2,S}$ 
 acts
on $f$ by the formula
$$
D^{2,S}f(\gamma)=S(\widetilde{D}^2f(\gamma)).
$$
\end{definition}

\begin{remark}
We use the space $H^1_{0,0}$ instead of $H^1_{0}$ for the definition of the second order derivative operator   because the directional derivative  $d_Xf$ of the section $f\in \mathcal {E}_m$ is  covariant if and only if $X(1)=0$ (see~\cite{LV2001}). 
\end{remark}

\section{Levy trace and Levy Laplacian}
\label{Sec4}

Let $T^2_{AGV}$  be the space of all continuous bilinear real-valued functionals on    $H_{0,0}^1\times H_{0,0}^1$ that have the form
\begin{multline}
\label{f''}
Q(u,v)=\int_0^1\int_0^1Q^V(t,s)<u(t),v(s)>dtds+\\
+\int_0^1Q^L(t)<u(t),v(t)>dt+
\\
+\frac 12\int_0^1Q^S(t)<\dot{u}(t),v(t)>dt+\frac 12\int_0^1Q^S(t)<\dot{v}(t),u(t)>dt,
\end{multline}
where
$Q^V\in L_2([0,1]\times[0,1],T^2(\mathbb{R}^4))$,
$Q^L\in L_1([0,1],Sym^2(\mathbb{R}^4))$,
$Q^S\in L_\infty([0,1],\Lambda^2(\mathbb{R}^4))$,
where $T^2(\mathbb{R}^4)$, $Sym^2(\mathbb{R}^4)$ and $\Lambda^2(\mathbb{R}^4))$ are the spaces of all
tensors, all symmetrical tensors  and all antisymmetrical tensors of type $(0,2)$ on $\mathbb{R}^4$ respectively.

In the fact, the space  of bilinear functionals $T^2_{AGV}$  was considered in  the paper~\cite{AGV1994} by Accardi, Gibilisco and Volovich (see also~\cite{AV1981}). The kernel $Q^V(\cdot,\cdot)$ is called the Volterra kernel,
the $Q^L(\cdot)$ is the Levy kernel and $Q^S(\cdot)$ is the singular kernel. 
These kernels are defined in a unique way (see~\cite{AGV1994}).

\begin{definition}
The Levy trace $tr^{AGV}_L$ acts on $Q\in T^2_{AGV}$ by the formula
$$
tr^{AGV}_LQ=\int_0^1tr\,Q^L(t)dt.
$$
\end{definition}

\begin{definition}
The Levy Laplacian  $\Delta^{AGV}_L$  is the second order differential operator  $D^{2,tr^{AGV}_L}$ 
associated with the Levy trace $tr^{AGV}_L$.
\end{definition}

This operator was introduced by Accardi, Gibilisco and Volovich in~\cite{AGV1994} for the flat case and  by Leandre and Volovich in~\cite{LV2001} for the case of Riemannian manifold.

\begin{example}
\label{example1}
Let $f_1\in C^\infty(M,\mathbb{R})$. Let
$L_{f_1}\colon \Omega_m\to \mathbb{R}$
be defined by:
$$
L_{f_1}(\gamma)=\int_0^1{f_1}(\gamma(t))dt.
$$
Then,
$$
 \Delta^{AGV}_L L_{f_1}(\gamma)=\int_0^1 \Delta_{(M,g)} f_1(\gamma(t))dt,
$$
where $\Delta_{(M,g)}$ is the Laplace-Beltrami operator on the manifold $M$.
\end{example}

Let us introduce the modification of the Levy trace that is connected with instantons. 
Let $W\in C^1([0,1],SO(4))$.
We can consider it as an orthogonal operator on $L_2([0,1],\mathbb{R}^4)$ defined by
$$
(Wu)(t)=W(t)u(t).
$$
The space $H_0^1([0,1],\mathbb{R}^4)$ is invariant under the action of  $W$. 

\begin{definition}
The modificated Levy trace is a linear functional $tr^W_L$ on  $T^2_{AGV}$  defined by
$$
tr^W_L Q=tr_L^{AGV}(W^\ast QW).
$$ 
\end{definition}

\begin{definition}
The modificated  Levy Laplacian $\Delta_L^W$ associated with the curve $W\in C^1([0,1],SO(4))$ is the second order differentional operator  $D^{2,tr^W_L}$.
\end{definition}

The groups $S^3_L$ and $S^3_R$ are the normal subgroups of $SO(4)$ that consists of real matrices of the form 
\begin{equation}
\begin{pmatrix}
a&-b&-c&-d\\
b&\;\,\, a&-d&\;\,\, c\\
c&\;\,\, d&\;\,\, a&-b\\
d&-c&\;\,\, b&\;\,\, a
\end{pmatrix}\text{ and }
\begin{pmatrix}
a&-b&-c&-d\\
b&\;\,\, a&\;\,\, d&-c\\
c&-d&\;\,\, a&\;\,\, b\\
d&\;\,\, c&-b&\;\,\, a
\end{pmatrix}
\end{equation}
respectively, where $a^2+b^2+c^2+d^2=1$. The Lie algebras  $Lie(S^3_L)$ and $Lie(S^3_R)$ of the Lie groups  $S^3_L$ and $S^3_R$ consist of real  matrices of the form   
$$\begin{pmatrix}
0&-b&-c&-d\\
b&\;\,\, 0&-d&\;\,\, c\\
c&\;\,\, d&\;\,\, 0&-b\\
d&-c&\;\,\, b&\;\,\, 0
\end{pmatrix}\text{ and }
\begin{pmatrix}
0&-b&-c&-d\\
b&\;\,\, 0&\;\,\, d&-c\\
c&-d&\;\,\, 0&\;\,\, b\\
d&\;\,\, c&-b&\;\,\, 0
\end{pmatrix}
$$
respectively. So it holds that $$so(4)=Lie(S^3_L)\oplus Lie(S^3_R)$$ and
$$Lie(S^3_L)\cong Lie(S^3_R)\cong so(3).$$ 
Let the symbols  $P_L$ and $P_R$ denote the orthogonal projections in $so(4)$ on the sub algebras  $Lie(S^3_L)$  and $Lie(S^3_R)$ respectively.

Let formally define the action of the Hodge star on $Q^S(t)$:
$$*Q^S_{\mu\nu}(t)=\frac 12 \sum_{\lambda=1}^4\sum_{\kappa=1}^4\epsilon_{\mu\nu\lambda\kappa}Q^S_{\lambda\kappa}(t).$$
Let  $Q^S_+(t)=\frac 12(Q^S(t)+*Q^S(t))$ and $Q^S_-(t)=\frac 12(Q^S(t)-*Q^S(t))$.
Due to the fact that $Q^S(t)$ is anti symmetric, it can be considered as an element from the algebra $so(4)$.
Let $Q^S_+(t)=P_L(Q^S(t))$ and $Q^S_-(t)=P_R(Q^S(t))$.

If $W\in C^1([0,1],SO(4))$ let $L_W(t)=W^{-1}(t)\dot{W}(t)$. Then $L_W$ is continuous 
curve in  $so(4)$:
\begin{equation}
\label{LW}
L_W^\ast(t)=-L_W(t).
\end{equation}

\begin{proposition}
Let $W\in C^1([0,1],SO(4))$.
It  holds
\begin{multline}
tr^{W}_LQ
=\int_0^1trQ^L(t)dt-\int_0^1tr(P_L(L_W(t))Q^S_+(t))dt-\\
-\int_0^1tr(P_R(L_W(t))Q^S_-(t))dt
\end{multline}
\end{proposition}
\begin{proof}
We have

\begin{multline}
\label{f''}
Q(Wu,Wv)=\int_0^1\int_0^1Q^V(t,s)<W(t)u(t),W(s)v(s)>dtds+\\
+\int_0^1Q^L(t)<W(t)u(t),W(t)v(t)>dt+
\\
+\frac 12\int_0^1Q^S(t)<W(t)\dot{u}(t)+\dot{W}(t)u(t),W(t)v(t)>dt+\\+\frac 12\int_0^1Q^S(t)<W(t)\dot{v}(t)+\dot{W}(t)v(t),W(t)u(t)>dt.
\end{multline}

Using direct calculations, we get that $Q_W:=W^\ast QW\in T^2_{AGV}$ and
\begin{multline}
\label{f''}
Q_W(u,v)=Q(Wu,Wv)=\int_0^1\int_0^1Q_W^V(t,s)<u(t),v(s)>dtds+\\
+\int_0^1Q_W^L(t)<u(t),v(t)>dt+
\\
+\frac 12\int_0^1Q^S_W(t)<\dot{u}(t),v(t)>dt+\frac 12\int_0^1Q^S_W(t)<\dot{v}(t),u(t)>dt,
\end{multline}
where the Volterra kernel of $Q_W$ has the form
$$
Q_W^V(t,s)=W^\ast (t)Q^V(t,s)W(s);
$$
the Levy kernel of $Q_W$  has the form
\begin{multline}
\label{LevykernQV}
Q_W^L(t)=\\=W^\ast (t)Q^L(t)W(t)+\frac 12W^\ast (t)L_W^\ast (t)Q^S(t)W(t)-\frac 12W^\ast (t)Q^S(t)L_W(t)W(t)
\end{multline}
and the singular kernel of $Q_S$ has the form
\begin{equation}
\label{SingkernQV}
Q^S_W(t)=W^\ast (t)Q^S(t)W(t).
\end{equation}

Due to $W(t)W^\ast (t)=W^\ast (t)W(t)=Id$ and equality~(\ref{LW}), we have
\begin{multline}
\label{eqQQw}
tr(Q_W^L(t))=\\=tr (W^\ast (t)Q^L(t)W(t))+\frac 12tr(W^\ast (t)L_W^\ast (t)Q^S(t)W(t))-\frac 12(W^\ast (t)Q^S(t)L_W(t)W(t))=\\=tr(W(t)W^\ast (t)Q^L(t))+\frac 12tr(W(t)W^\ast (t)L_W^\ast (t)Q^S(t))-\frac 12(W(t)W^\ast (t)Q^S(t)L_W(t))=\\=
tr Q^L(t)-tr(L_W(t)Q^S(t))=\\=tr Q^L(t)-tr(P_L(L_W(t))Q^S_+(t))-tr(P_R(L_W(t))Q^S_-(t)).
\end{multline}
The last equality holds due to $tr(P_L(L_W(t))Q^S_-(t))=tr(P_R(L_W(t))Q^S_+(t))=0$.
Equality~(\ref{eqQQw}) implies the statement of the proposition.
\end{proof}

\section{Value of Levy Laplacian on parallel transport}
\label{Sec5}

The first derivative of the parallel transport is well-known.

\begin{proposition}
\label{prop1deriv}
The first derivative of the parallel  transport has the form
\begin{multline}
\label{firstderpar}
d_XU_{t_2,t_1}(\gamma)=-\int_{t_1}^{t_2}U_{t_2,t}(\gamma)F(\gamma(t))<X(t),\dot{\gamma}(t)>U_{t,t_1}(\gamma)dt-\\
-A(\gamma(t_2))X(t_2)U_{t_2,t_1}(\gamma)
+U_{t_2,t_1}(\gamma)A(\gamma(t_1))X(t_1).
\end{multline}
\end{proposition}
\begin{proof}
For the proof see~\cite{Driver} (see also~\cite{Gross,Volkov2019a}).
\end{proof}

\begin{remark}
If $X(t_1)=X(t_2)=0$, formula~(\ref{firstderpar}) has an interpretation as Non-Abelean Stokes formula (see~\cite{Arefieva1980,Gross}). If $h\in H^1_{0,0}$, the first derivative of the parallel  transport has the form
\begin{multline}
\label{firstderpar1}
\widetilde{D}U_{1,0}(\gamma)h=d_{\widetilde{h}}U_{1,0}(\gamma)=-\int_{0}^{1}U_{1,t}(\gamma)F(\gamma(t))<\widetilde{h}(\gamma,t),\dot{\gamma}(t)>U_{t,0}(\gamma)dt=\\
=-\int_{0}^{1}U_{1,t}(\gamma)F(\gamma(t))<Z_\mu(\gamma,t),\dot{\gamma}(t)> h^\mu(t)U_{t,0}(\gamma)dt.
\end{multline}
\end{remark}

Let $h_1,h_2\in H^1_{0,0}$. The second derivative of the parallel transport has the form
\begin{multline}
\label{der2}
<\widetilde{D}^2U_{1,0}(\gamma)h_1,h_2>=d_{\widetilde{h}_2}\widetilde{D}U_{1,0}(\gamma)h_2=\\
=-\int_0^1\int_0^s  U_{1,t}(\gamma)F(\gamma(t))<\dot{\gamma}(t),\widetilde{h}_1(\gamma,t)>\times \\ \times U_{t,s}(\gamma)F(\gamma(s))<\dot{\gamma}(s), \widetilde{h}_2(\gamma,s)>U_{s,0}(\gamma)dsdt-\\-
\int_0^1\int_0^s U_{1,s}(\gamma)F(\gamma(s))<\dot{\gamma}(s),\widetilde{h}_2(\gamma,s)>\times \\ \times U_{s,t}(\gamma)F(\gamma(t))<\dot{\gamma}(t),\widetilde{h}_1(\gamma,t)>U_{t,0}(\gamma)dtds-\\-
\int_0^1U_{1,t}(\gamma)F(\gamma(t))<\widetilde{h}_1(\gamma,t) ,\frac d{dt}\widetilde{h}_2(\gamma,t)>U_{t,0}(\gamma)dt-\\
-\int_0^1U_{1,t}(\gamma)F(\gamma(t))<d_{\widetilde{h}_2}\widetilde{h}_1(\gamma,t),\dot{\gamma}(t)>U_{t,0}(\gamma)dt-\\-
\int_0^1U_{1,t}(\gamma)\partial_{\widetilde{h}_2(t)}F(\gamma(t))<\widetilde{h}_1(\gamma,t),\dot{\gamma}(t)>U_{t,0}(\gamma)dt-\\
-\int_0^1U_{1,t}(\gamma)[A(\gamma(t))\widetilde{h}_2(\gamma,t),F(\gamma(t))<\widetilde{h}_1(\gamma,t),\dot{\gamma}(t)>]U_{t,0}(\gamma)dt.
\end{multline}
Note that
\begin{multline}
\frac d{dt}\widetilde{h}_2(\gamma,t)=\frac d{dt}(Q_{t,0}(\gamma)h_2(t))=(\frac d{dt}Q_{t,0}(\gamma))h_2(t)+Q_{t,0}(\gamma)\dot{h}_2(t)=\\
=\Gamma(\gamma(t))<\dot{\gamma}(t),\widetilde{h}_2(\gamma,t)>+Q_{t,0}(\gamma)\dot{h}_2(t)
\end{multline}
and
\begin{multline}
d_{\widetilde{h}_2}(\widetilde{h}_1(\gamma;t))=d_{\widetilde{h}_2}Q_{t,0}(\gamma)h_1(t)=\\
=-\int_{0}^{t}Q_{t,s}(\gamma)R(\gamma(s))<Q_{s,0}(\gamma)h_1(t),\widetilde{h_2}(\gamma;s),\dot{\gamma}(s)>ds-\\-\Gamma(\gamma(t))<\widetilde{h}_1(\gamma;t),\widetilde{h}_2(\gamma;t)>,
\end{multline}
where $R=(R^\lambda_{\mu \nu \kappa})$ is a Riemannian curvature tensor. Let 
$$k^v(\gamma;t,s)<h_1(t),h_2(s)>=Q_{t,s}(\gamma)R(\gamma(s))<Q_{s,0}(\gamma)h_1(t),\widetilde{h_2}(\gamma;s),\dot{\gamma}(s)>.$$
After we group the terms,
we get
\begin{multline}
\label{der21}
<\widetilde{D}^2U_{1,0}(\gamma)h_1,h_2>=\\
=\int_0^1\int_0^1 K^V(\gamma;s,t)<h_1(s),h_2(t)>dtds-\\
-\int_0^1U_{1,t}(\gamma)\nabla_{\widetilde{h}_2(t)}F(\gamma(t))<\widetilde{h}_1(\gamma,t),\dot{\gamma}(t)>U_{t,0}(\gamma)dt-\\
-\int_0^1U_{1,t}(\gamma)F(\gamma(t))<\widetilde{h}_1(\gamma,t),\widetilde{\dot{h}}_2(\gamma,t)>U_{t,0}(\gamma)dt,
\end{multline}
where
\begin{multline}
\label{KV}
K^V(\gamma;t,s)<h_1(t),h_2(s)>=\\=\begin{cases}
U_{1,t}(\gamma)F(\gamma(t))<\dot{\gamma}(t),\widetilde{h}_1(\gamma,t)> U_{t,s}(\gamma)F(\gamma(s))
<\dot{\gamma}(s),\widetilde{h}_2(\gamma,s)>U_{s,0}(\gamma)+\\+
U_{1,t}(\gamma)F(\gamma(t))<k^v(\gamma;t,s)<h_1(t),h_2(s)>,\dot{\gamma}(t)>U_{t,0}(\gamma)
,&\text{if $t\geq s$}\\
U_{1,s}(\gamma)F(\gamma(s))<\dot{\gamma}(s),\widetilde{h}_2(\gamma,s)>U_{s,t}(\gamma) F(\gamma(t))<\dot{\gamma}(t),\widetilde{h}_1(\gamma,t)>U_{t,0}(\gamma)
,&\text{if $t<s$}.
\end{cases}
\end{multline}
It is possible to transform the last two term in~(\ref{der21})
by integrating  the expression
$$
\frac 12\int_0^1U_{1,t}(\gamma)F(\gamma(t))<\widetilde{h}_1(\gamma,t),\widetilde{\dot{h}}_2(\gamma,t)>U_{t,0}(\gamma)dt
$$
by parts and using the Bianchi identities
\begin{multline*}
\nabla_{\widetilde{h}_2(t)}  F(\gamma(t))<\widetilde{h}_1(t),\dot{\gamma}(t)>+\nabla_{\widetilde{h}_1(t)} F(\gamma(t))<\dot{\gamma}(t),\widetilde{h}_2(t)>+\\+\nabla_{\dot{\gamma}(t)} F(\gamma(t))<\widetilde{h}_2(t),\widetilde{h}_1(t)>=0.
\end{multline*}
Due to $h_1(0)=h_1(1)=h_2(0)=h_2(1)=0$, we have
\begin{multline}
\label{der23}
<\widetilde{D}^2U_{1,0}(\gamma)h_1,h_2>=\\
=\int_0^1\int_0^1 K^V(\gamma;t,s)<h_1(t),h_2(s)>dtds+\\
+\int_0^1K^L(\gamma,t)<h_1(t),h_2(t)>dt+\\
+\frac 12\int_0^1K^S(\gamma;t)<\dot{h}_1(t),h_2(t)>dt+\frac 12\int_0^1K^S(\gamma;t)<\dot{h}_2(t),h_1(t)>dt,
\end{multline}
where the Levy kernel $K^L$ and the singular kernel $K^S$  have  the form
\begin{multline*}
K^L_{\mu\nu}(\gamma;t)=\frac 12U_{1,t}(\gamma)(-\nabla_{Z_\mu(\gamma,t)} F(\gamma(t))<Z_\nu(\gamma,t),\dot{\gamma}(t)>-\\-\nabla_{Z_\nu(\gamma,t)} F(\gamma(t))<Z_\mu(\gamma,t),\dot{\gamma}(t)>)U_{t,0}(\gamma),
\end{multline*}
and
$$
K^S_{\mu\nu}(\gamma;t)=U_{1,t}(\gamma)F(\gamma(t))<Z_\mu(\gamma,t),Z_\nu(\gamma,t)>U_{t,0}(\gamma)
$$
respectively in the orthonormal basis  $\{Z_1(\gamma,t),Z_2(\gamma,t),Z_3(\gamma,t),Z_4(\gamma,t)\}$.

Bellow, if $\gamma\in \Omega_m$ and $K(\gamma,\cdot)$ is a section in the pullback  bundle  $\gamma^{*}\Lambda^{2}\otimes adP$,
the symbol $\textbf{K}(\gamma,t)$ means that we consider $K(\gamma,t)$ in the orthonormal basis $\{Z_1(\gamma,t),\ldots,Z_4(\gamma,t)\}$.
So  $\textbf{K}_{\mu\nu}(\gamma,t)=K(\gamma,t)<Z_\mu(\gamma,t),Z_\nu(\gamma,t)>$ is antisimmetrical matrix which elements are  $su(N)$-matrices.

So we have proved the following theorem.
\begin{theorem}
The value of the modificated  Levy Laplacian $\Delta_L^W$ on  the  parallel transport is
\begin{multline}
\label{lllaplace0}
\Delta_L^WU_{1,0}(\gamma)=-\int_0^1U_{1,t}(\gamma) D_A^\ast F(\gamma(t))\dot{\gamma}(t)U_{t,0}(\gamma)dt-\\
-\int_0^1U_{1,t}(\gamma)tr (L_W(t) {\bf F}(\gamma(t)))U_{t,0}(\gamma)dt=\\
=-\int_0^1U_{1,t}(\gamma) D_A^\ast F(\gamma(t))\dot{\gamma}(t)U_{t,0}(\gamma)dt-\\
-\int_0^1U_{1,t}(\gamma)tr(P_L(L_W(t)) {\bf F}_{+}(\gamma(t)))U_{t,0}(\gamma)dt-\\-\int_0^1U_{1,t}(\gamma)tr(P_R(L_W(t)) {\bf F}_{-}(\gamma(t)))U_{t,0}(\gamma)dt.
\end{multline}
\end{theorem}
If $F_+=0$, then $D_A^\ast F$. In this case, the first  and the second terms in the right side of~(\ref{lllaplace0}) are equal to zero.
So, if additionally $W\in C^1([0,1],S^3_R)$,  we have $P_R(L_W(t))=0$ and, hence, 
$$
\Delta_L^WU_{1,0}=0.
$$ 
In the following section, we will prove that the converse statement  is  true in some sense.

\section{Main theorem}
\label{Sec6}

In the beginning, we prove  auxiliary lemmas.
\begin{lemma}
\label{mainlemma}
If the parallel transport $U_{1,0}$ is a solution of the equation
$$\Delta^W_LU_{1,0}=0,$$ then
the connection $A$ is 
a solution of the Yang-Mills equations:
$$D_A^\ast F=0.$$
\end{lemma}

\begin{proof} 
For any    $\gamma\in \Omega_m$ let  $\gamma^r\in \Omega_m$ be defined as 
\begin{equation}
\label{gammar}
\gamma^r(t)=\begin{cases}
\gamma(t), &\text{if $t\leq r$,}\\
\gamma(r), &\text{if $t>r$.}
\end{cases}
\end{equation}
Then the properties of the parallel transport~(\ref{reparametr}),~(\ref{group}) and~(\ref{indent}) imply
\begin{multline*}
\Delta^W_LU_{1,0}(\gamma^r)=-\int_0^rU_{r,t}(\gamma)D_A^\ast F(\gamma(t))\dot{\gamma}(t)U_{t,0}(\gamma)dt-\\
-\int_0^rU_{r,t}(\gamma)tr (L_W(t){\bf F}(\gamma(t)))U_{t,0}(\gamma)dt-\\
-\int_r^1 tr (L_W(t){\bf F}(\gamma(r)))dt U_{r,0}(\gamma).
\end{multline*}

Assume that $\gamma\in C^1([0,1],M)$. Let us introduce the function $J\in C^1([0,1],\mathcal{L}(E_m,E_{\gamma(1)}))$ in the following way:
\begin{multline*}
J(r)=U_{1,r}(\gamma)(\Delta^W_LU_{1,0}(\gamma^r))=\\=-
\int_0^rU_{1,t}(\gamma)D_A^\ast F(\gamma(t))\dot{\gamma}(t)U_{t,0}(\gamma)dt-\\
-\int_0^rU_{1,t}(\gamma)tr (L_W(t){\bf F}(\gamma(t)))U_{t,0}(\gamma)dt-\\
-U_{1,r}(\gamma)\int_r^1 tr (L_W(t){\bf F}(\gamma(r)))dt U_{r,0}(\gamma).
\end{multline*}
Due to 
$$
\frac {d}{dr}U_{1,r}(\gamma){\bf F}(\gamma(r))U_{r,0}(\gamma)=U_{1,r}(\gamma)\nabla_{\dot{\gamma}(r)}{\bf F}(\gamma(r))U_{r,0}(\gamma),
$$
we have
\begin{multline*}
J'(r)=-U_{1,r}(\gamma)D_A^\ast F(\gamma(t))\dot{\gamma}(t)U_{r,0}(\gamma)+\\
-U_{1,r}(\gamma)\int_r^1tr( L_W(t)\nabla_{\dot{\gamma}(r)}{\bf F}(\gamma(r))dt U_{r,0}(\gamma).
\end{multline*}
The equality $\Delta^W_LU_{1,0}=0$ implies $J\equiv0$ and, hence, $J'\equiv 0$.
Then $$J'(1)=-D_A^\ast F(\gamma(1))\dot{\gamma}(1)U_{1,0}(\gamma)=0.$$  Hence, for any $\gamma\in C^1([0,1],M)$  with the origin 
at  $m$ we have
$$
D_A^\ast F(\gamma(1))\dot{\gamma}(1)=0.
$$
Taking suitable curves $\gamma$, we find that the connection $A$
is a solution of the Yang--Mills equations on  $M$.
\end{proof}

\begin{lemma}
\label{importlemma}
Let ${\bf a}\in so(4)$ and $W_{\bf a}(t)=e^{{\bf a}t}$.
The parallel transport $U_{1,0}$ is a solution of the Laplace equation for  the modificated Levy Laplacian
$$\Delta^{W_{\bf a}}_LU_{1,0}=0,$$ if and only if the connection $A$ is a solution of the Yang-Mills equations
and  for any $\gamma\in \Omega_m$ and   $r\in [0,1]$   the following holds
$$
tr({\bf a}{\bf F}(\gamma(r)))=0.
$$
\end{lemma}
\begin{proof}
Let $\Delta^{W_{\bf a}}_LU_{1,0}=0$.
Consider an arbitrary $\gamma\in \Omega_m$. Let 
$\gamma_r\in \Omega_m$ be defined as~(\ref{gammar}) in the proof of Lemma~\ref{mainlemma}.
For  $\varepsilon>0$ let us introduce $\gamma_{r,\varepsilon}\in \Omega_x$ by the following way
$$\gamma_{r,\varepsilon}(t)=\begin{cases}
\gamma({rt}/{\varepsilon}), &\text{if $0 <t\leq \varepsilon$,}\\
\gamma(r),  &\text{if $\varepsilon<t\leq1$.}\\
\end{cases}
$$
Then the properties of the parallel transport~(\ref{paraltrans}) and~(\ref{group}) imply
\begin{multline}
\label{forproof}
\Delta^{W_{\bf a}}_LU_{1,0}(\gamma_{r,\varepsilon})=\\=\int_{\varepsilon}^1tr({\bf a}{\bf F}(\gamma_r(1)))dtU_{1,0}(\gamma_r)+
\int_{0}^{\varepsilon}U_{1,\frac {t}{\varepsilon}}(\gamma_r)tr({\bf a}{\bf F}\left(\gamma_r\left(t /\varepsilon \right)\right))U_{\frac {t}{\varepsilon},0}(\gamma_r)dt=0.
\end{multline}
Let us introduce the notations
 $$I_1(\varepsilon)=U_{1,0}^{-1}(\gamma_r)\int_{\varepsilon}^1tr({\bf a}{\bf F}(\gamma_r(1)))dtU_{1,0}(\gamma_r)$$ and
$$
I_2(\varepsilon)=\int_{0}^{\varepsilon}U_{\frac {t}{\varepsilon},0}^{-1}(\gamma_r)tr({\bf a}{\bf F}\left(\gamma_r\left(t/\varepsilon\right)\right))U_{\frac {t}{\varepsilon},0}(\gamma_r)dt.
$$
Due to equality~(\ref{forproof}), we have
$$
I_1(\varepsilon)=-I_2(\varepsilon).
$$
Let  $\|\cdot\|$ be a standard norm on $so(4)$. 
The mapping $[0,1]\ni r\to \|tr({\bf a}{\bf F}(\gamma_x(r))\|$ is continuous. Hence, there exists $C>0$, such that $$\sup_{r\in[0,1]}\|tr({\bf a}{\bf F}(\gamma_x(r))\|\leq C.$$
So we have the estimate
\begin{equation}
\|I_1(\varepsilon)\|=\|I_2(\varepsilon)\|\leq
\int_{0}^{\varepsilon}\|tr({\bf a}{\bf F}\left(\gamma_x\left(t/\varepsilon\right)\right)\|dt\leq
 C\varepsilon.
\end{equation}
From this estimate it follows that 
\begin{equation*}
U_{1,0}^{-1}(\gamma_r)\int_{0}^1tr({\bf a}{\bf F}(\gamma_r(1)))dtU_{1,0}(\gamma_r)=\lim_{\varepsilon \to 0}I_1(\varepsilon)=0.
\end{equation*}
Hence,
\begin{equation}
tr({\bf a}{\bf F}(\gamma(r)))=tr({\bf a}{\bf F}(\gamma_r(1)))
=\int_0^1tr({\bf a}{\bf F}(\gamma_r(1))dt=0.
\end{equation}
The other side of the statement of the lemma is trivial.

\end{proof}

\begin{theorem} 
Let $\{\bf{e_1},\bf{e_2},\bf{e_3}\}$ be  a basis of the Lie algebra $Lie(S^3_L)$ (Lie algebra $Lie(S^3_R)$). 
Let $W_i(t)=e^{t\bf{e_i}}$ for $i\in \{1,2,3\}$.
The following two assertions are equivalent:
\begin{enumerate}
\item  a connection $A$ is a solution of anti-selfduality equations (selfduality) equations~(\ref{autodual});
\item  $\Delta^{W_i}_LU_{1,0}=0$ for $i\in\{1,2,3\}$.
\end{enumerate}
\end{theorem}
\begin{proof}
Let $\{\bf{e_1},\bf{e_2},\bf{e_3}\}$ be a some  basis of $Lie(S^3_L)$ and let  $\Delta^{W_i}_LU_{1,0}=0$ for $i\in\{1,2,3\}$
Then 
Lemma~\ref{importlemma} implies that for any $\gamma\in \Omega_m$ and   $r\in [0,1]$   the following holds
$$
tr({\bf e_i}{\bf F}(\gamma(r)))=tr(P_L({\bf e_i}){\bf F_+}(\gamma(r)))+tr(P_R({\bf e_i}){\bf F_-}(\gamma(r)))=tr({\bf e_i}{\bf F_+}(\gamma(r)))=0
$$
for $i\in \{1,2,3\}$.
Then for any $\gamma\in \Omega_m$ and   $r\in [0,1]$ we have  ${\bf F_+}(\gamma(r))=0$.
Hence, $A$ is an instanton. The other side of the statement of the theorem is trivial.
\end{proof}

\section{Conclusion}

In the paper, we introduced a class of the modificated L\'evy Laplacians parameterized by the choice of a curve in the group $SO(4)$ on the infinite dimensional manifold.
We  showed  that it is possible to choose three Laplacians from this class such that  a connection  on the  4-dimensional orientable Riemannian manifold is an instanton  if and only if the  parallel transport  associated with this connection is a solution of the 
Laplace equations for these  Laplacians.

\section*{Acknowledgments}

The author would like to express his deep gratitude to L.~Accardi, O.~G.~Smolyanov and I.~V.~Volovich for helpful discussions.

This work is supported by the Russian Science Foundation under grant 19-11-00320.

\end{document}